\documentclass[10pt]{IEEEtran}
\IEEEoverridecommandlockouts

\usepackage[numbers,sort&compress]{natbib}
\def\BibTeX{{\rm B\kern-.05em{\sc i\kern-.025em b}\kern-.08em T\kern-.1667em\lower.7ex\hbox{E}\kern-.125emX}}

\usepackage{amsmath,amssymb,amsfonts}

\usepackage{graphicx}
\usepackage{booktabs}
\usepackage{multirow}
\usepackage{adjustbox}
\usepackage{stfloats}
\usepackage{enumitem}

\usepackage{algorithm}
\usepackage{algpseudocode}

\usepackage{hyperref}
\usepackage{url}

\usepackage{textcomp}
\usepackage{nicefrac}
\usepackage{microtype}
\usepackage{xcolor}

\usepackage{setspace}
\usepackage{tikz}

\usepackage{orcidlink}
\usepackage{comment}
\usepackage{braket}
\usepackage[skip=0pt, font=footnotesize]{caption}
\usepackage{multicol}
\usepackage{multirow}
\usepackage[most]{tcolorbox}
\usepackage{amsthm}
\usepackage{stfloats}

\newtcbtheorem[number within=section]{securitytheorem}{\footnotesize Theorem}%
{
    enhanced,
    breakable,
    colback=blue!3,
    colframe=blue!45!black,
    coltitle=black,
    fonttitle=\bfseries,
    boxed title style={
        colback=blue!12,
        colframe=blue!45!black
    },
    attach boxed title to top left={xshift=2mm,yshift=-2mm},
    top=2mm,
    bottom=2mm,
    left=2mm,
    right=2mm,
    arc=1mm,
    boxrule=0.5pt
}{thm}

\newtcbtheorem[number within=section]{securitycorollary}{\footnotesize Corollary}%
{
    enhanced,
    breakable,
    colback=gray!4,
    colframe=gray!60!black,
    coltitle=black,
    fonttitle=\bfseries,
    boxed title style={
        colback=gray!15,
        colframe=gray!60!black
    },
    attach boxed title to top left={xshift=2mm,yshift=-2mm},
    top=2mm,
    bottom=2mm,
    left=2mm,
    right=2mm,
    arc=1mm,
    boxrule=0.5pt
}{cor}
\usepackage{pifont}
\newcommand{\cmark}{\ding{51}} \newcommand{\xmark}{\ding{55}} \newcommand{\pmark}{$\triangle$}
\begin{document}

\bstctlcite{IEEEexample:BSTcontrol} 

\title{
QIAP: A Quantum-Inspired Authentication Protocol with BB84-Based Channel Monitoring
}

\author{\IEEEauthorblockN{Owais Ishtiaq Siddiqui\textsuperscript{1}, Nouhaila Innan\textsuperscript{2,3}, Hira Zaheer\textsuperscript{1}, Muhammad Shafique\textsuperscript{2,3}
}\\ 
\IEEEauthorblockA{
\textsuperscript{1}Department of Computer Science, Scuola IMT Alti Studi Lucca, Italy\\
\textsuperscript{2}eBRAIN Lab, Division of Engineering, New York University Abu Dhabi (NYUAD), Abu Dhabi, UAE\\ \textsuperscript{3}Center for Quantum and Topological Systems (CQTS), NYUAD Research Institute, NYUAD, Abu Dhabi, UAE\\
\{owaisishtiaq.siddiqui, hira.zaheer \}@imtlucca.it, \{nouhaila.innan, muhammad.shafique\}@nyu.edu}}

\maketitle
\thispagestyle{empty}
\pagestyle{empty}
\begin{abstract}
Authentication in quantum-aware communication settings requires mechanisms that verify user legitimacy while also assessing whether the communication channel exhibits excessive disturbance. This paper presents QIAP, a stateful hybrid framework that combines classical commitment-based challenge-response authentication with BB84-based quantum-channel monitoring. QIAP uses a dual-branch verification structure inspired by rectilinear and diagonal basis selection. In the rectilinear branch, the prover opens a hash-based commitment and the verifier checks consistency with the registered round token. In the diagonal branch, the prover returns an HMAC tag bound to the verifier nonce, round identifier, and submitted commitment without disclosing the round-specific sub-secret. Under the stated hash, HMAC, freshness, secure-registration, protected-storage, and high-entropy-secret assumptions, QIAP provides computational resistance to the modeled wrong-secret and replay attempts and protects the long-term master secret from direct recovery. After successful authentication, the BB84 module estimates the quantum bit error rate (QBER) and aborts the session when the measured value reaches or exceeds the configured threshold. In the reported simulation, no successful authentication was observed among 1,000 tested wrong-secret attempts per branch. The baseline and full intercept-resend conditions produced mean QBER values of approximately 0.02 and 0.25, respectively. The classical authentication stage required approximately 0.6 ms, and the implemented authentication messages required approximately 113 bytes per round. These results indicate that QIAP can provide a lightweight hybrid framework for combining secret-based identity verification with QBER-based channel-disturbance assessment under the stated threat model.
\end{abstract}


\section{Introduction}

Authentication is essential in secure communication systems, distributed networks, and privacy-preserving services. A user or device must prove possession of valid authentication material without exposing its long-term secret. Classical authentication can verify identity, but it does not determine whether a communication channel has been intercepted or disturbed. Reused transcripts, weak credentials, and poor session binding can also enable impersonation, replay, and man-in-the-middle attacks.

Zero-knowledge proofs (ZKPs) support privacy-preserving verification by allowing a prover to demonstrate that a statement is valid without revealing its private witness \cite{xing2025zero,goldwasser1989knowledge}. ZKP protocols are commonly defined through completeness, soundness, and zero-knowledge properties \cite{peng2026survey}, and have been applied to digital identity, blockchain systems, password authentication, and privacy-preserving computation \cite{chi2023privacy,kuznetsov2024enhanced,kara2023password,xing2023zero,xing2025zero}. However, classical authentication and ZKP-based methods verify the prover rather than the physical condition of the communication channel.

Quantum communication provides a separate security function. In BB84 quantum key distribution, an adversary who measures transmitted quantum states may disturb them and increase the quantum bit error rate (QBER) \cite{lee2022eavesdropping,imran2024quantum,BB84}. BB84 can therefore detect channel disturbance, while classical authentication determines whether a party possesses valid authentication material. These two functions address different security goals.

Recent quantum authentication methods use quantum measurements or QBER behavior during identity verification \cite{garcia2024experimental}. Such methods may require quantum hardware for authentication. Classical challenge-response methods are easier to deploy, but they cannot directly measure quantum-channel disturbance. This motivates a hybrid design that performs identity authentication using classical cryptographic operations and validates the quantum channel before accepting the session.

This paper proposes the Quantum-Inspired Authentication Protocol (QIAP), a stateful framework combining classical commitment-based challenge-response authentication with BB84-based channel validation. The term ``quantum-inspired'' refers to the use of rectilinear and diagonal authentication branches modeled after the two BB84 bases. The authentication layer uses classical hash and HMAC operations and is not presented as a formal zero-knowledge proof.

For each authentication round, the prover derives a single-use sub-secret from a high-entropy master secret and creates a nonce-bound commitment using fresh randomness. After receiving the commitment, the verifier selects one of two branches. In the rectilinear branch, the prover opens the commitment and the verifier checks the registered token and commitment consistency. In the diagonal branch, the prover returns an HMAC tag bound to the nonce, round identifier, and commitment without revealing the sub-secret. After successful authentication, the BB84 stage estimates the QBER and aborts the session when the selected threshold is reached or exceeded.

QIAP separates identity acceptance from channel acceptance. Each branch applies hash, commitment-opening, or HMAC verification, while multiple rounds provide independent single-use authentication records and prevent reuse of previously challenged material.
\textbf{The main contributions of this work are as follows:}
\begin{itemize}
    \item We propose QIAP, a dual-branch classical authentication framework combined with BB84-based channel validation.

    \item We design nonce-bound authentication rounds using hash commitments, HMAC binding tags, fresh verifier nonces, and single-use verification records.

    \item We provide a security analysis covering completeness, computational impersonation resistance, replay resistance, master-secret privacy, and QBER-based session abort.

    \item We observe no successful wrong-secret authentication in 1,000 attempts per branch, corresponding to an observed rejection rate of \(100\%\) under the tested implementation.

    \item We obtain mean QBER values of approximately \(0.02\) for the clean channel and \(0.25\) for the intercept--resend condition using \(1024\) simulated qubits and a threshold of \(\theta=0.11\).

    \item We measure approximately \(0.6\,\mathrm{ms}\) for the classical authentication stage and approximately \(113\) bytes of communication per authentication round.
\end{itemize}

\section{Background and Related Work}

This section summarizes the concepts underlying QIAP and positions the proposed design relative to classical privacy-preserving authentication and quantum-channel validation.

\subsection{Classical Zero-Knowledge Proofs}

A zero-knowledge proof (ZKP) is an interactive or non-interactive protocol in which a prover convinces a verifier that a statement is true without revealing the private witness supporting that statement \cite{goldwasser1989knowledge}. In a typical interactive protocol, the prover sends an initial commitment, the verifier issues a challenge, and the prover returns a response, as illustrated in Fig.~\ref{fig:classical_zkp}.

\begin{figure}[ht]
    \centering
    \includegraphics[width=0.48\textwidth]{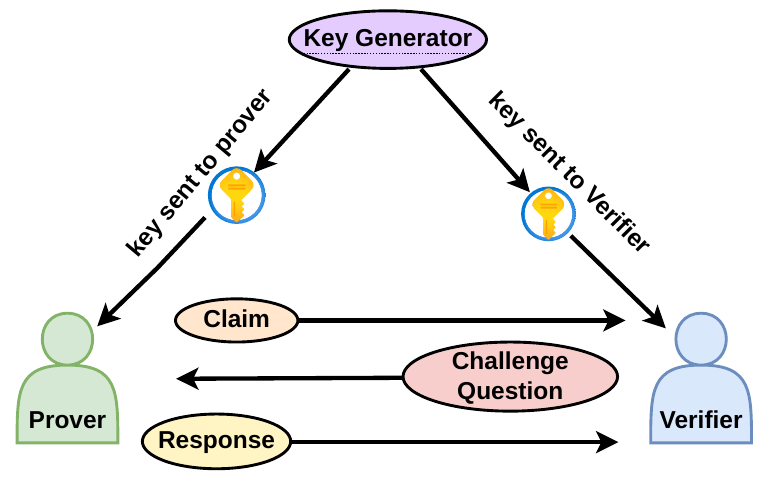}
    \caption{General interaction in a classical zero-knowledge proof. The prover uses a private witness to answer a verifier challenge without revealing the witness itself.}
    \label{fig:classical_zkp}
\end{figure}

Let \(L\) be a language in NP, \(x\) a public statement, and \(w\) a witness satisfying \(R(x,w)=1\). A proof system \((P,V)\) is commonly characterized by three properties. \emph{Completeness} requires an honest prover with a valid witness to be accepted with overwhelming probability. \emph{Soundness} requires a prover without a valid witness to be accepted only with negligible probability. \emph{Zero knowledge} requires the verifier's view to be computationally indistinguishable from the output of a simulator that does not know \(w\) \cite{9520375}.

These definitions provide useful principles for privacy-preserving authentication. However, QIAP does not claim the formal zero-knowledge property because its rectilinear branch reveals a round-specific sub-secret. QIAP instead aims to protect the long-term master secret while verifying possession of secret-derived authentication material.

\subsection{BB84-Based Channel Validation}

BB84 is a foundational quantum key distribution (QKD) protocol based on state preparation and measurement in rectilinear and diagonal bases \cite{BB84}. If an adversary measures a transmitted state using an incompatible basis, the measurement may disturb the state and create errors in the sifted key.

The disturbance is estimated using the quantum bit error rate (QBER):
\begin{equation}
    \mathrm{QBER}
    =
    \frac{N_{\mathrm{mismatch}}}
         {N_{\mathrm{sifted}}},
    \label{eq:qber}
\end{equation}
where \(N_{\mathrm{mismatch}}\) is the number of mismatched sifted bits and \(N_{\mathrm{sifted}}\) is the number of sifted bits selected for comparison.

In the simulation used in this work, the channel is accepted when
\begin{equation}
    \mathrm{QBER}<\theta,
    \qquad
    \theta=0.11,
    \label{eq:qber-threshold}
\end{equation}
and the session is aborted otherwise. The value \(\theta=0.11\) is used as a simulation decision threshold rather than a universal deployment limit. A deployed QKD system must also account for finite-key effects, device imperfections, error-correction leakage, privacy amplification, losses, and its selected security model.

BB84 and authentication provide different security functions. Authentication verifies possession of valid identity-related material, whereas QBER monitoring evaluates whether the quantum channel exhibits excessive disturbance. QKD also requires an authenticated classical channel during basis reconciliation and post-processing. QIAP therefore uses the classical authentication result to authorize and bind the subsequent BB84 session.

\subsection{Research Gap and Related Approaches}

Existing methods usually emphasize one of three directions. Classical ZKP-based authentication protects identity-related witnesses but does not directly measure physical disturbance on a quantum channel \cite{szczegielniak2025zero,chriki2026enhancing}. Quantum authentication protocols use quantum states to verify users or devices, but generally require quantum communication resources during authentication \cite{garcia2024experimental,lou2026novel}. Broader quantum-cryptography frameworks support channel-security and other quantum functions, but are not necessarily designed as lightweight classical identity-authentication mechanisms \cite{bozzio2025quantum,adil2025quantum}.

QIAP addresses this design gap by combining a classical, stateful challenge-response layer with a separate BB84 validation stage. Table~\ref{tab:related} compares the design categories using properties that can be stated without treating one scheme's internal verification rule as a universal benchmark.

\begin{table*}[t] \centering \caption{Qualitative comparison of QIAP with recent authentication and quantum-security approaches.} 
\label{tab:related} 
\setlength{\tabcolsep}{2pt} 
\footnotesize 
\begin{tabular}{@{}p{0.8cm}p{7cm}cccc@{}} \toprule \textbf{Work} & \textbf{Approach} & \textbf{Classical Auth.} & \textbf{Channel Detection} & \textbf{Secret Privacy} & \textbf{Dual Verification} \\ 
\midrule 
\cite{lou2026novel}& Quantum mutual authentication and key agreement & \pmark & \cmark & \cmark & \xmark 
\\ \cite{szczegielniak2025zero}& Zero-knowledge authentication for 5G and beyond & \cmark & \xmark & \cmark & \xmark \\ 
\cite{chaudhry2025zk} & Quantum-resistant authentication using ZK proofs & \cmark & \xmark & \cmark & \xmark 
\\ \cite{bozzio2025quantum} & Quantum cryptography beyond QKD & \xmark & \cmark & \cmark & \xmark \\ \cite{garcia2024experimental} & Experimental quantum ZKP user authentication & \xmark & \cmark & \cmark & \xmark \\ 
\midrule \textbf{QIAP} & \textbf{Dual-branch authentication with BB84 QBER monitoring} & \textbf{\cmark} & \textbf{\cmark} & \textbf{\cmark} & \textbf{\cmark} 
\\ 
\bottomrule 
\multicolumn{6}{l}{\scriptsize \cmark: supported;\quad \xmark: not supported;\quad \pmark: partially supported.} 
\end{tabular} 
\end{table*}
QIAP performs identity authentication using classical cryptographic operations and invokes quantum hardware only for the BB84 stage. Its distinguishing feature is the explicit separation of identity acceptance and channel acceptance within one session flow, with both stages required for session acceptance.

\section{Methodology}

QIAP combines a classical commitment-based authentication layer with a BB84-based channel-validation layer. The authentication layer verifies possession of valid secret-derived material, while the BB84 layer evaluates whether the quantum channel exhibits excessive disturbance before a communication session is accepted. QIAP consists of three phases: registration, basis-dependent authentication, and BB84-based channel validation. Fig.~\ref{fig:qiap_architecture} illustrates the protocol flow.

\begin{figure}[ht]
    \centering
    \includegraphics[width=\linewidth]{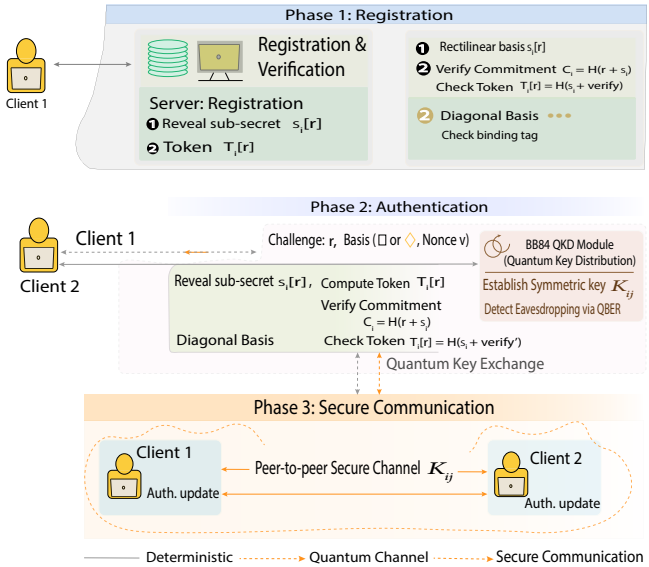}
    \caption{QIAP architecture with registration, basis-dependent authentication, and BB84-based channel validation before session acceptance.}
    \label{fig:qiap_architecture}
\end{figure}

The notation used throughout the protocol is summarized as follows. The prover's high-entropy master secret is denoted by \(s\), the authentication-round index by \(r\), the round-derived sub-secret by \(s_r\), fresh prover randomness by \(\rho_r\), the fresh verifier nonce by \(\nu_r\), the commitment by \(C_r\), the rectilinear verification token by \(T_r\), the diagonal verification key by \(K_r^{(d)}\), the branch challenge by \(b_r\), and the diagonal-branch HMAC tag by \(\tau_r\). All concatenated fields are encoded using an unambiguous, length-delimited binary encoding.

\subsection{Registration Phase}

Let \(R\) be the number of registered authentication rounds. For each
\(r\in\{0,\ldots,R-1\}\), the prover derives the round-specific sub-secret
\begin{equation}
    s_r
    =
    H\!\left(
        s \parallel \operatorname{enc}(r)
    \right),
    \label{eq:round-secret}
\end{equation}
where \(H(\cdot)\) is a cryptographic hash function and
\(\operatorname{enc}(r)\) is an unambiguous binary encoding of the round index.

Two branch-specific verification values are then derived from \(s_r\). The rectilinear verification token is
\begin{equation}
    T_r
    =
    H\!\left(
        s_r \parallel \mathtt{verify}
    \right),
    \label{eq:rect-token}
\end{equation}
and the diagonal verification key is
\begin{equation}
    K_r^{(d)}
    =
    \operatorname{HMAC}_{s_r}
    \!\left(
        \mathtt{key\_diag}
    \right).
    \label{eq:diag-key}
\end{equation}

The verifier stores the single-use round record
\begin{equation}
    \mathcal{L}_r
    =
    \left(
        T_r,
        K_r^{(d)},
        \mathsf{unused}
    \right),
    \label{eq:round-record}
\end{equation}
rather than storing the master secret \(s\). Each record is marked as consumed after one authentication attempt, regardless of whether the attempt succeeds or fails. This rule prevents reuse of disclosed or previously challenged round material. The verifier storage is assumed to be protected because knowledge of \(K_r^{(d)}\) permits generation of valid diagonal-branch tags for that round. New round records must be provisioned through the authenticated registration process before the available pool is exhausted.

\subsection{Commitment and Challenge Phase}

For an unused round \(r\), the verifier samples a fresh nonce \(\nu_r\) and sends \((r,\nu_r)\) to the prover. The prover recomputes \(s_r\), samples fresh randomness \(\rho_r\), and constructs the nonce-bound commitment
\begin{equation}
    C_r
    =
    H\!\left(
        \rho_r
        \parallel s_r
        \parallel \nu_r
        \parallel \operatorname{enc}(r)
    \right).
    \label{eq:commitment}
\end{equation}

After receiving \(C_r\), the verifier samples the branch challenge
\begin{equation}
    b_r
    \xleftarrow{\$}
    \left\{
        \mathsf{Rect},
        \mathsf{Diag}
    \right\},
    \label{eq:branch-challenge}
\end{equation}
and sends \(b_r\) to the prover. Sampling the branch after receiving the commitment prevents the prover from constructing \(C_r\) for only one selected response type.

\subsection{Rectilinear Verification Branch}

If \(b_r=\mathsf{Rect}\), the prover opens the commitment by sending
\begin{equation}
    \operatorname{Resp}^{(\mathsf{Rect})}_r
    =
    \left(
        s_r,\rho_r
    \right).
    \label{eq:rect-response}
\end{equation}

The verifier accepts the round only if
\begin{equation}
\begin{aligned}
    H\!\left(
        s_r \parallel \mathtt{verify}
    \right)
    &=
    T_r,
    \\
    H\!\left(
        \rho_r
        \parallel s_r
        \parallel \nu_r
        \parallel \operatorname{enc}(r)
    \right)
    &=
    C_r.
\end{aligned}
\label{eq:rect-checks}
\end{equation}

The first condition checks consistency with the registered round token. The second condition checks that the response opens the commitment for the current nonce and round index. This branch reveals \(s_r\), but not the master secret \(s\).

\subsection{Diagonal Verification Branch}

If \(b_r=\mathsf{Diag}\), the prover does not reveal \(s_r\). Instead, it computes
\begin{equation}
    \tau_r
    =
    \operatorname{HMAC}_{K_r^{(d)}}
    \!\left(
        \nu_r
        \parallel \operatorname{enc}(r)
        \parallel C_r
    \right),
    \label{eq:diag-tag}
\end{equation}
and sends
\begin{equation}
    \operatorname{Resp}^{(\mathsf{Diag})}_r
    =
    \tau_r.
    \label{eq:diag-response}
\end{equation}

The verifier recomputes the expected tag using the stored
\(K_r^{(d)}\), the current nonce \(\nu_r\), the round index \(r\), and the received commitment \(C_r\). The round is accepted only if the received and recomputed tags are equal. This branch verifies freshness and commitment binding without disclosing \(s_r\).

After either branch is evaluated, the verifier marks \(\mathcal{L}_r\) as consumed. A consumed record cannot be used in a later authentication attempt. Because failed attempts also consume records, implementations should apply rate limiting and authenticated session initiation to reduce record-exhaustion attacks.

\subsection{BB84 Channel Validation}

The BB84 phase is executed only after all required authentication rounds are accepted. Using \(N_q\) transmitted qubits, the module estimates the QBER according to Eq. \eqref{eq:qber}. The channel decision is
\begin{equation}
    \mathsf{SessionOK}
    =
    \begin{cases}
        1, & \mathrm{QBER}<\theta,\\
        0, & \mathrm{QBER}\geq\theta,
    \end{cases}
    \label{eq:session-decision}
\end{equation}
where \(\theta\) is the selected QBER threshold.

If \(\mathsf{SessionOK}=1\), BB84 post-processing proceeds with sifting, error correction, and privacy amplification to derive the session key \(K_{ij}\). If \(\mathsf{SessionOK}=0\), the protocol aborts and accepts no session key. The classical BB84 messages are assumed to be authenticated and bound to the completed QIAP authentication exchange so that an adversary cannot substitute or splice a different quantum session.

\section{Mathematical Security Analysis}

This section analyzes QIAP under the stated threat model, with the protocol flow summarized in Algorithm~\ref{alg:qiap}. The analysis covers honest execution, impersonation resistance, replay resistance, protection of the master secret, and QBER-based session abort. The claims are computational and depend on the assumptions listed below. 
\subsection{Security Assumptions}

QIAP relies on the following assumptions:
\begin{itemize}
    \item \textbf{Hash security:} The hash function \(H\) is preimage resistant and second-preimage resistant.

    \item \textbf{HMAC security:} HMAC is existentially unforgeable under chosen-message attack for parties that do not know the HMAC key.

    \item \textbf{High-entropy master secret:} The master secret \(s\) is sampled from a sufficiently large key space and is not derived directly from a human-memorable password.

    \item \textbf{Freshness:} Each verifier nonce \(\nu_r\) and prover randomness value \(\rho_r\) is sampled independently and is not reused.

    \item \textbf{Unambiguous encoding:} All concatenated protocol fields use fixed-length or length-delimited binary encodings.

    \item \textbf{Single-use records:} Each verifier record \(\mathcal{L}_r\) is consumed after one authentication attempt.

    \item \textbf{Protected verifier storage:} The stored values \(T_r\) and \(K_r^{(d)}\) are protected from disclosure and unauthorized modification.

    \item \textbf{Secure registration:} Registration is authenticated and integrity protected.

    \item \textbf{Authenticated BB84 classical channel:} Classical messages used for basis reconciliation and post-processing are authenticated and bound to the completed QIAP authentication exchange.
\end{itemize}

\begin{algorithm}[t]
\footnotesize
\caption{QIAP protocol flow}
\label{alg:qiap}
\begin{algorithmic}[1]
\Require Master secret \(s\), registered rounds \(R\), BB84 length \(N_q\), threshold \(\theta\)
\Ensure Authentication decision \(\mathsf{AuthOK}\) and session decision \(\mathsf{SessionOK}\)

\State \textbf{Registration}
\For{\(r=0,\ldots,R-1\)}
    \State Derive \(s_r\), \(T_r\), and \(K_r^{(d)}\)
    \State Store \(\mathcal{L}_r=(T_r,K_r^{(d)},\mathsf{unused})\)
\EndFor

\State \textbf{Authentication}
\State \(\mathsf{AuthOK}\gets 0\)
\For{\(r=0,\ldots,R-1\)}
    \If{\(\mathcal{L}_r\) is consumed}
        \State \Return \((0,0)\)
    \EndIf

    \State Verifier samples fresh \(\nu_r\) and sends \((r,\nu_r)\)
    \State Prover samples fresh \(\rho_r\), computes \(C_r\), and sends \(C_r\)
    \State Verifier samples \(b_r\in\{\mathsf{Rect},\mathsf{Diag}\}\) and sends \(b_r\)

    \If{\(b_r=\mathsf{Rect}\)}
        \State Prover sends \((s_r,\rho_r)\)
        \State Verify both conditions in Eq. \eqref{eq:rect-checks}
    \Else
        \State Prover computes and sends \(\tau_r\)
        \State Verify \(\tau_r\) using Eq. \eqref{eq:diag-tag}
    \EndIf

    \State Mark \(\mathcal{L}_r\) as consumed
    \If{round verification fails}
        \State \Return \((0,0)\)
    \EndIf
\EndFor

\State \(\mathsf{AuthOK}\gets 1\)

\State \textbf{BB84 channel validation}
\State Run BB84 using \(N_q\) qubits and estimate QBER
\If{\(\mathrm{QBER}<\theta\)}
    \State Perform BB84 post-processing and derive \(K_{ij}\)
    \State \(\mathsf{SessionOK}\gets 1\)
\Else
    \State Abort and set \(\mathsf{SessionOK}\gets 0\)
\EndIf

\State \Return \((\mathsf{AuthOK},\mathsf{SessionOK})\)
\end{algorithmic}
\end{algorithm}
\subsection{Authentication Completeness}

\begin{securitytheorem}{\footnotesize  Authentication Completeness}{auth-completeness}
If the prover and verifier follow QIAP honestly and use the same registered record \(\mathcal{L}_r\), then the authentication round is accepted in either verification branch.
\end{securitytheorem}

\begin{proof}
In the rectilinear branch, the honest prover derives the registered sub-secret \(s_r\) using Eq. \eqref{eq:round-secret} and opens \(C_r\) using the same values \(s_r\), \(\rho_r\), \(\nu_r\), and \(r\) used in Eq. \eqref{eq:commitment}. Therefore, both checks in Eq. \eqref{eq:rect-checks} hold.

In the diagonal branch, the prover and verifier use the same value \(K_r^{(d)}\), nonce \(\nu_r\), round index \(r\), and commitment \(C_r\). They therefore compute the same tag according to Eq. \eqref{eq:diag-tag}. Hence, an honestly executed round is accepted in either branch.
\end{proof}

\subsection{Impersonation Resistance}

\begin{securitytheorem}{\footnotesize Rectilinear-Branch Impersonation Resistance}{rect-soundness}
Assume the adversary does not know the registered round sub-secret \(s_r\). The probability that the adversary is accepted in the rectilinear branch is negligible under the stated hash assumptions.
\end{securitytheorem}

\begin{proof}
Acceptance requires a pair \((\widehat{s}_r,\widehat{\rho}_r)\) satisfying
\begin{equation}
    H\!\left(
        \widehat{s}_r \parallel \mathtt{verify}
    \right)
    =
    T_r
    \label{eq:rect-token-adversary}
\end{equation}
and
\begin{equation}
    H\!\left(
        \widehat{\rho}_r
        \parallel \widehat{s}_r
        \parallel \nu_r
        \parallel \operatorname{enc}(r)
    \right)
    =
    C_r.
    \label{eq:rect-commitment-adversary}
\end{equation}

If \(\widehat{s}_r\neq s_r\), then satisfying Eq. \eqref{eq:rect-token-adversary} requires finding a second preimage for the registered token input or otherwise violating the assumed security of \(H\). If the adversary uses the genuine \(s_r\), then it already possesses the required round secret. In addition, the submitted opening must match the commitment bound to the fresh nonce and round index. Under the stated assumptions, an adversary without \(s_r\) satisfies both checks only with negligible probability.
\end{proof}

\begin{securitytheorem}{\footnotesize  Diagonal-Branch Impersonation Resistance}{diag-soundness}
Assume the adversary does not know \(K_r^{(d)}\). The probability that the adversary is accepted in the diagonal branch on a fresh message is negligible under HMAC unforgeability.
\end{securitytheorem}

\begin{proof}
Acceptance requires a valid tag for
\begin{equation}
    m_r
    =
    \nu_r
    \parallel \operatorname{enc}(r)
    \parallel C_r.
    \label{eq:diag-message}
\end{equation}
Because \(\nu_r\) is fresh and the record is single-use, \(m_r\) is fresh except with negligible nonce-collision probability. Producing a valid tag for \(m_r\) without \(K_r^{(d)}\) constitutes an HMAC forgery. Therefore, acceptance occurs only with negligible probability.
\end{proof}

\begin{securitycorollary}{ \footnotesize Per-Branch Authentication Security}{auth-soundness}
For either selected branch, an adversary lacking the branch-required secret material is accepted only with negligible probability under the stated assumptions.
\end{securitycorollary}

The corollary is a per-branch statement. QIAP does not claim a soundness error of \(2^{-R}\), and the analysis does not treat repeated rounds as the sole source of authentication security.

\subsection{Replay Resistance}

\begin{securitytheorem}{\footnotesize  Replay Resistance}{replay-resistance}
A transcript accepted in one authentication attempt cannot be replayed successfully in a later attempt except with negligible probability, provided nonces are fresh and round records are not reused.
\end{securitytheorem}

\begin{proof}
In the rectilinear branch, \(C_r\) is bound to \(\nu_r\) and \(r\) by Eq. \eqref{eq:commitment}. A captured commitment and opening therefore fail under a new nonce. Reuse under the same round index is also rejected because \(\mathcal{L}_r\) has been consumed.

In the diagonal branch, \(\tau_r\) is bound to \(\nu_r\), \(r\), and \(C_r\). A captured tag does not verify for a new nonce or commitment. Producing a valid replacement tag without \(K_r^{(d)}\) would violate HMAC unforgeability. Hence, replay succeeds only if freshness or record-consumption assumptions fail, or with negligible cryptographic probability.
\end{proof}

\subsection{Master-Secret Protection}

\begin{securitytheorem}{\footnotesize  Master-Secret Protection}{master-secret-privacy}
Assume \(s\) is sampled from a high-entropy key space and \(H\) is preimage resistant. Observing protocol transcripts, including disclosed round sub-secrets from rectilinear rounds, does not enable efficient recovery of \(s\).
\end{securitytheorem}

\begin{proof}
A rectilinear transcript may reveal
\begin{equation}
    s_r
    =
    H\!\left(
        s \parallel \operatorname{enc}(r)
    \right).
\end{equation}
Recovering \(s\) from one or more disclosed values \(s_r\) requires either inverting the hash computation or searching the master-secret space. Because \(s\) is assumed to have high entropy, exhaustive search is computationally infeasible. The remaining transcript values are derived from \(s_r\), fresh randomness, nonces, commitments, or HMAC computations and do not provide a more direct recovery method under the stated assumptions.
\end{proof}

This theorem protects only the master secret. It does not claim zero knowledge, because the rectilinear branch discloses the round-specific value \(s_r\). It also does not apply when \(s\) is a low-entropy password, since disclosed \(s_r\) values would permit offline testing of password guesses.

\subsection{Verifier-Storage Limitation}

\begin{securitytheorem}{\footnotesize  Effect of Diagonal-Key Disclosure}{diag-key-disclosure}
If an adversary obtains an unused stored key \(K_r^{(d)}\), then the adversary can generate a valid diagonal-branch response for that round.
\end{securitytheorem}

\begin{proof}
The diagonal response is computed directly as
\[
    \tau_r
    =
    \operatorname{HMAC}_{K_r^{(d)}}
    \!\left(
        \nu_r
        \parallel \operatorname{enc}(r)
        \parallel C_r
    \right).
\]
Knowledge of \(K_r^{(d)}\) is therefore sufficient to compute a valid tag for any selected nonce and commitment for the same unused record.
\end{proof}

This result is why protected verifier storage is part of the threat model. QIAP does not provide verifier-database-compromise resistance for the diagonal branch.

\subsection{BB84 Channel Validation}

\begin{securitytheorem}{\footnotesize Intercept--Resend Disturbance}{qber-detection}
Under the idealized BB84 intercept--resend model, interception of transmitted qubits increases the expected QBER relative to an otherwise identical clean channel.
\end{securitytheorem}

\begin{proof}
The sender and receiver independently select rectilinear or diagonal bases. An intercepting adversary does not know the sender's basis and may measure in the incompatible basis. Such a measurement disturbs the transmitted state. When the receiver later measures in the sender's basis, the disturbance can produce a mismatched sifted bit. Repeated interception therefore increases the expected mismatch rate and shifts the QBER distribution upward relative to the clean-channel case.
\end{proof}

The theorem applies to the modeled intercept--resend attack. Elevated QBER indicates channel disturbance but does not, by itself, identify eavesdropping as the cause.

\begin{securitytheorem}{\footnotesize  QBER-Based Session Abort}{qber-abort}
If the measured QBER satisfies \(\mathrm{QBER}\geq\theta\), QIAP sets \(\mathsf{SessionOK}=0\) and accepts no session key.
\end{securitytheorem}

\begin{proof}
The result follows directly from the decision rule in Eq. \eqref{eq:session-decision}. QIAP accepts the BB84 stage only when \(\mathrm{QBER}<\theta\). Therefore, any measured value satisfying \(\mathrm{QBER}\geq\theta\) causes session abort.
\end{proof}

\subsection{Security Statement}

\begin{securitytheorem}{\footnotesize  Conditional Security Properties of QIAP}{combined-security}
Under the assumptions of hash security, HMAC unforgeability, high master-secret entropy, fresh randomness, fresh nonces, unambiguous encoding, single-use records, protected verifier storage, secure registration, and authenticated BB84 classical communication, QIAP provides authentication completeness, per-branch computational impersonation resistance, replay resistance, master-secret protection, and threshold-based session abort under the stated threat model.
\end{securitytheorem}

\begin{proof}
Authentication completeness follows from Theorem~\ref{thm:auth-completeness}. Rectilinear- and diagonal-branch impersonation resistance follow from Theorems~\ref{thm:rect-soundness} and~\ref{thm:diag-soundness}. Replay resistance follows from Theorem~\ref{thm:replay-resistance}. Master-secret protection follows from Theorem~\ref{thm:master-secret-privacy}. QBER-based session abort follows from Theorem~\ref{thm:qber-abort}. These properties remain conditional on protected verifier storage, as shown by Theorem~\ref{thm:diag-key-disclosure}, and on the remaining assumptions listed above.
\end{proof}

\section{Results and Discussion}

\subsection{Experimental Setup}

QIAP was evaluated in a hybrid classical--quantum simulation covering wrong-secret authentication, BB84 channel disturbance, noise sensitivity, execution time, and communication cost. The classical layer used SHA-256 for round-secret derivation, verification tokens, and commitments, and HMAC-SHA-256 for diagonal-branch tags. Each authentication session used \(R=20\) single-use round records, a fresh \(128\)-bit verifier nonce, fresh \(256\)-bit prover randomness, and a \(256\)-bit HMAC tag.

Each BB84 trial simulated \(N_q=1024\) single-qubit transmissions under clean and full intercept--resend conditions. Depolarizing noise was varied from \(p_e=0\) to \(0.20\) in steps of \(0.01\). The session threshold was fixed at \(\theta=0.11\), and the session was aborted when \(\mathrm{QBER}\geq\theta\).

Wrong-secret authentication was tested over 1,000 independent attempts per branch. Clean and intercept--resend QBER behavior was evaluated over \(500\) trials per condition, while noise sensitivity was evaluated over \(300\) trials for each \(p_e\). Fixed pseudorandom seeds were used for reproducibility. Execution time was measured as software wall-clock time, and communication cost was computed from the transmitted protocol fields.

The implementation used Python, NumPy, Matplotlib, \texttt{hashlib}, \texttt{hmac}, and \texttt{secrets}. Experiments were executed on an AMD Ryzen 5 PRO 4650U processor with \(16\)~GB RAM under Ubuntu~22.04 through Windows Subsystem for Linux.

\subsection{Wrong-Secret Authentication}

QIAP rejected all 1,000 wrong-secret attempts in each branch, producing an observed rejection rate of \(100\%\) under the tested implementation. No false rejection was observed for honest authentication attempts.

In the rectilinear branch, an incorrect \(s_r\) fails the registered-token check and cannot produce a valid opening of the nonce-bound commitment. In the diagonal branch, an incorrect secret derives a different key and produces an invalid HMAC tag.

With zero successful impersonations in 1,000 trials, the rule-of-three approximation gives an upper \(95\%\) confidence bound of approximately
\begin{equation}
    p_{\mathrm{imp}}
    <
    \frac{3}{1000}
    =
    0.003
\end{equation}
for the tested attack generator. This bound characterizes the observed attack success rate in the experiment and is distinct from the cryptographic forgery probability of SHA-256 or HMAC-SHA-256.

\subsection{BB84 Channel Disturbance}

\begin{figure}[t]
    \centering
    \includegraphics[width=1\linewidth]{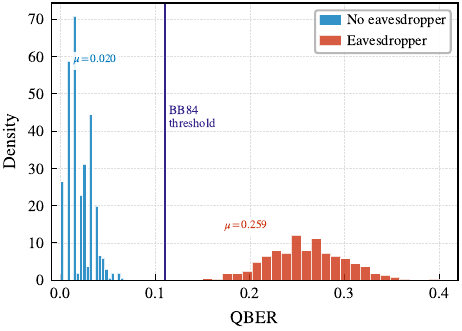}
    \caption{Simulated QBER distributions for clean and full intercept--resend channels. The clean condition is centered near \(\mu\approx0.02\), while the intercept--resend condition is centered near \(\mu\approx0.25\).}
    \label{fig:qber_discrimination}
\end{figure}

As shown in Fig.~\ref{fig:qber_discrimination}, the clean and intercept--resend conditions produced well-separated QBER distributions. The clean channel had a mean QBER of approximately \(0.02\), whereas the full intercept--resend condition had a mean near \(0.25\). With \(\theta=0.11\), the two tested distributions were separated by a clear margin.

This result is consistent with the disturbance expected from full intercept--resend measurement in BB84. The observed behavior corresponds to the simulated full-interception setting with the selected qubit count and QBER threshold.

\subsection{Execution Time}

\begin{figure}[t]
    \centering
    \includegraphics[width=1\linewidth]{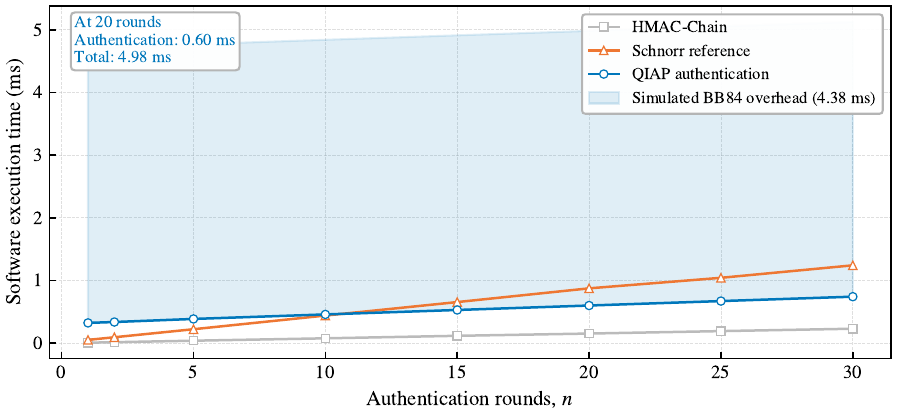}
    \caption{Software execution-time comparison among HMAC-chain, Schnorr authentication, and QIAP. The QIAP authentication stage remains below \(1\)~ms, while the simulated BB84 stage adds separate processing time.}
    \label{fig:authentication_latency}
\end{figure}

As shown in Fig.~\ref{fig:authentication_latency}, the classical QIAP authentication stage required approximately \(0.6\)~ms on the evaluated software and hardware platform. Including the simulated BB84 stage increased the total execution time to approximately \(5\)~ms. These values measure local software execution and do not represent optical transmission delay, detector delay, reconciliation cost, or hardware QKD latency.

The HMAC-chain baseline was faster because it performed fewer cryptographic operations. Schnorr authentication represents a public-key reference with a different security and computation profile. The comparison therefore focuses on software execution cost across the evaluated implementations.

\subsection{Communication Cost}

\begin{figure}[ht]
    \centering
    \includegraphics[width=1\linewidth]{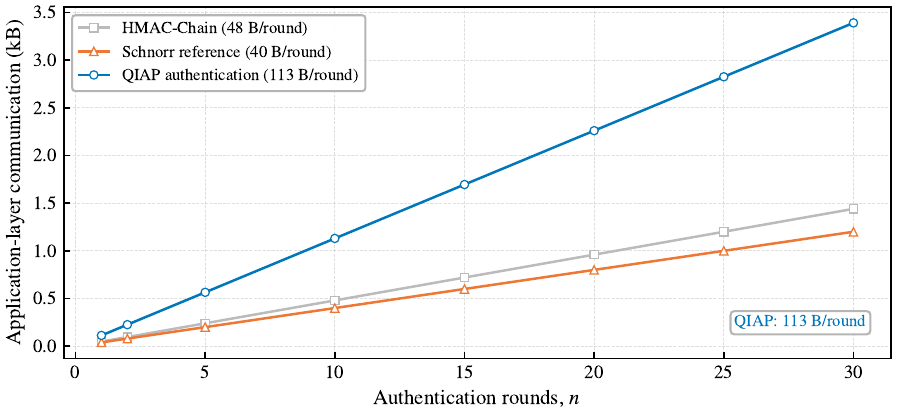}
    \caption{Communication cost as a function of the number of authentication rounds. QIAP grows linearly at approximately \(113\) bytes per round in the implemented message format.}
    \label{fig:communication_overhead}
\end{figure}

\begin{figure*}[ht]
    \centering
    \includegraphics[width=\linewidth]{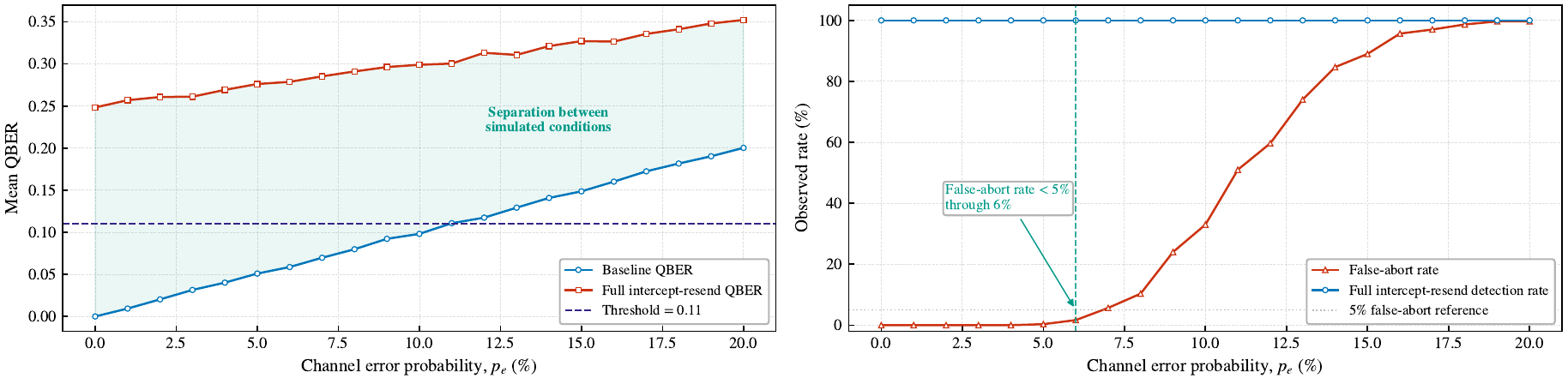}
    \caption{Effect of depolarizing error probability \(p_e\) on simulated BB84 behavior: (a) QBER under clean and intercept--resend conditions and (b) disturbance-detection and false-abort rates for \(\theta=0.11\).}
    \label{fig:qber_decision_rates}
\end{figure*}
The transmitted cryptographic fields have branch-dependent sizes. The fixed fields are the \(4\)-byte round identifier, \(16\)-byte nonce, \(32\)-byte commitment, and \(1\)-byte branch challenge. The rectilinear response contains a \(32\)-byte sub-secret and \(32\)-byte randomness value, while the diagonal response contains a \(32\)-byte HMAC tag. Hence, the raw cryptographic payloads are
\begin{align}
    B_{\mathsf{Rect}}^{\mathrm{raw}}
    &=
    4+16+32+1+32+32
    =
    117~\text{bytes},
    \\
    B_{\mathsf{Diag}}^{\mathrm{raw}}
    &=
    4+16+32+1+32
    =
    85~\text{bytes}.
\end{align}

For uniformly selected branches, the expected raw payload is 

\begin{equation}
      \mathbb{E}\!\left[B^{\mathrm{raw}}\right]
    =
    \frac{1}{2}B_{\mathsf{Rect}}^{\mathrm{raw}}
    +
    \frac{1}{2}B_{\mathsf{Diag}}^{\mathrm{raw}}
    =
    101~\text{bytes per round}.
\end{equation}

As shown in Fig.~\ref{fig:communication_overhead}, the communication cost increases linearly with the number of authentication rounds. The implemented message format adds approximately \(12\) bytes per round for serialization and framing metadata, including field identifiers, message-length information, and branch-dependent encoding. The resulting measured communication cost is therefore $    B_{\mathrm{impl}}
    \approx
    101+12
    =
    113~\text{bytes per round}.$

At this rate, \(30\) authentication rounds require approximately $    30\times113
    =
    3390~\text{bytes}
    \approx
    3.39~\text{kB}.$
This value excludes lower-layer network headers and BB84 quantum-transmission cost. It represents the application-layer authentication messages generated by the reported implementation.

\subsection{Noise Sensitivity and Decision Rates}

As shown in Fig.~\ref{fig:qber_decision_rates}, clean-channel QBER increased with \(p_e\), while the intercept--resend condition remained higher across the tested range. The disturbance-detection rate remained close to \(100\%\), whereas the false-abort rate increased as the clean-channel QBER approached \(\theta\).
Under the simulated configuration, false aborts remained low up to approximately \(p_e=0.06\) and increased beyond that point. This value defines the empirical operating boundary under the selected simulation model.

\subsection{Quantum-Channel Operating Region}

\begin{figure}[htpb]
    \centering
    \includegraphics[width=1\linewidth]{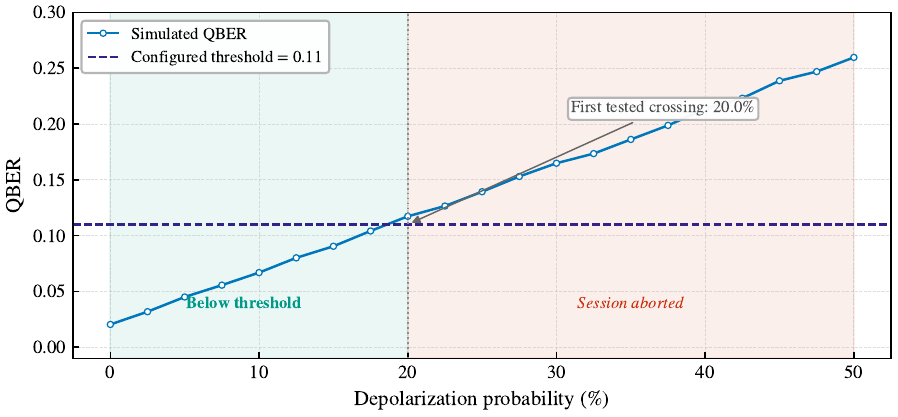}
    \caption{Simulated BB84 operating region. The session is accepted when \(\mathrm{QBER}<0.11\) and aborted otherwise.}
    \label{fig:quantum_operating_region}
\end{figure}

The threshold rule used by QIAP defines the BB84 operating region: the channel is accepted when \(\mathrm{QBER}<0.11\) and the session is aborted when \(\mathrm{QBER}\geq0.11\), as shown in Fig.~\ref{fig:quantum_operating_region}. For this operating-region analysis, the depolarizing error probability is extended beyond the primary evaluation range to identify the threshold-crossing region. Elevated QBER indicates increased channel disturbance but does not determine its source. Such disturbance may result from interception, device errors, calibration drift, transmission loss, or environmental noise. QIAP therefore uses threshold violation as a session-abort criterion without attributing the observed disturbance to a specific cause.
\subsection{Discussion}

The results show that QIAP separates classical identity authentication from quantum-channel validation while requiring both stages to succeed before session acceptance. The classical layer rejected all tested wrong-secret attempts, and the BB84 layer distinguished the clean and full intercept--resend conditions under the selected simulation settings.

The numerical findings reflect the implemented threat model, channel model, qubit count, threshold, and software environment. Extensions to partial interception, detector attacks, source imperfections, finite-key effects, and hardware execution remain open directions.

QIAP introduces higher execution and communication cost than the classical baselines because it uses branch-dependent authentication evidence and a separate BB84 stage. QIAP is designed for systems that require both secret-based identity verification and quantum-channel disturbance monitoring. Deployment also requires protected verifier storage, authenticated BB84 classical communication, record-pool management, and implementation-specific threshold selection.

\section{Conclusion}

This paper presented QIAP, a stateful hybrid framework that combines classical dual-branch authentication with BB84-based channel validation. The classical layer verifies possession of secret-derived authentication material, while the BB84 layer evaluates quantum-channel disturbance through QBER monitoring. A communication session is accepted only when both authentication and channel validation succeed.
In the experiments, QIAP rejected all 1,000 wrong-secret attempts per branch, producing an observed rejection rate of \(100\%\) under the tested implementation. The BB84 simulation produced mean QBER values of approximately \(0.02\) for the clean channel and \(0.25\) for the full intercept--resend condition using \(N_q=1024\) qubits and a threshold of \(\theta=0.11\). The classical authentication stage required approximately \(0.6\)~ms, and the implemented message format required approximately \(113\) bytes per authentication round.
These results show that QIAP can combine secret-based identity verification with quantum-channel disturbance monitoring at low software cost. The findings are established under the stated cryptographic assumptions and evaluated channel models. Elevated QBER serves as the criterion for session abort without attributing the observed disturbance to a specific cause.

Future work will examine hardware QKD integration, verifier-database-compromise resistance, record-pool management, lower communication cost, and stronger privacy-preserving authentication designs with formal security proofs.

\section*{Acknowledgment}
 This work was supported in part by the NYUAD Center for Quantum and Topological Systems (CQTS), funded by Tamkeen under the NYUAD Research Institute grant CG008, and the Center for Cyber Security (CCS), funded by Tamkeen under the NYUAD Research Institute Award G1104.
\bibliographystyle{IEEEtran}


{\scriptsize
\bibliography{refs}}
\end{document}